\documentclass[11pt]{article}
\usepackage{amsmath, amsthm, amscd, amsfonts, amssymb, graphicx, color}
\usepackage[utf8]{inputenc}
\usepackage{graphicx}
\usepackage{hyperref}
\usepackage[titlenumbered,ruled,noend,algosection]{algorithm2e}

\newtheorem{thm}{Theorem}[section]
\newtheorem{cor}[thm]{Corollary}

\DontPrintSemicolon

\begin{document}
\title{Time-Windowed Contiguous Hotspot Queries}
\author{Ali Gholami Rudi\thanks{
{Department of Electrical and Computer Engineering},
{Bobol Noshirvani University of Technology}, {Babol, Iran}.
Email: {\tt gholamirudi@nit.ac.ir}.}}
\date{}
\maketitle

\begin{abstract}
A hotspot of a moving entity is a region in which it spends a
significant amount of time.
Given the location of a moving object through a certain time
interval, i.e.~its trajectory, our goal is to find its hotspots.
We consider axis-parallel square hotspots of fixed side length,
which contain the longest contiguous portion of the trajectory.
Gudmundsson, van Kreveld, and Staals (2013)
presented an algorithm to find a hotspot of a trajectory in $O(n \log n)$,
in which $n$ is the number of vertices of the trajectory.
We present an algorithm for answering \emph{time-windowed} hotspot queries,
to find a hotspot in any given time interval.
The algorithm has an approximation factor of $1/2$ and
answers each query with the time complexity $O(\log^2 n)$.
The time complexity of the preprocessing step of the algorithm is $O(n)$.
When the query contains the whole trajectory,
it implies an $O(n)$ algorithm for finding approximate
contiguous hotspots.
\end{abstract}

\section{Introduction}
\label{sint}
The identification of popular places or hotspots is an important
preprocessing step for analyzing trajectories
(see \cite{zheng15} for several interesting applications)
and the recent growth of trajectory data sets is increasing the
emphasis on efficient trajectory analysis algorithms,
hotspot identification algorithms not excepted.

Several heuristics have been proposed for identifying hotspots
(e.g.~\cite{li08,yuan11,yuan15}).
To pave the way for devising accurate geometric algorithms for
identifying hotspots,
some authors have formalized the definition of hotspots and
posed several questions about their identification \cite{benkert10,gudmundsson13}.
One of the problems introduced by Gudmundsson et al.~is that of finding a
placement of an axis-aligned square of fixed side length,
which maximizes the time the entity spends inside it without
leaving it \cite{gudmundsson13}.  In other words, the square should contain a
contiguous sub-trajectory with the maximum duration.
They also present an algorithm for this problem with the
time complexity $O(n \log n)$.
Their algorithm relies on the fact that there exists a
hotspot with at least one trajectory vertex on its boundary.
It finds the hotspot by testing
the squares, on one of whose boundaries lies a vertex of the
sub-trajectory.

We study the problem of answering \emph{time-windowed}
hotspot queries.  Each query specifies a time interval and the
goal is finding a hotspot for the sub-trajectory of this interval
(e.g.~for finding the stay point of a bird in August or
finding the hotspot of a mobile device per day, week, and month).
Time-windowed geometric queries have recently attracted much attention.
Bannister et al.~introduce a framework for answering time-windowed
queries and present algorithms for answering such queries
for three problems including convex hull, for which they present an
algorithm that answers each query in poly-logarithmic time and performs the
preprocessing in $O(n\log n)$ time \cite{bannister14}.
For the time-windowed closest pair of points problem,
Chan and Pratt \cite{chan15} present a quadtree-based algorithm in the word-RAM model.
For decision problems, Bokal et al.~present time-windowed algorithms
for deciding hereditary properties (i.e., if a sequence has
the property, all of its subsequences do as well) \cite{bokal15}.
Chan and Prat \cite{chan16} improve Bokal et al.'s results by reducing time-windowed
decision problems to range successor problem and using dynamic data
structures.
None of these results, however, can be used directly for answering
hotspot queries.  Since its goal is finding a hotspot, the problem
cannot be stated as a decision problem and the techniques used for
time-windowed problems like convex hull do not seem applicable to hotspot identification.

In this paper we present an approximation algorithm for answering
time-windowed hotspot queries.  It answers each query with the time
complexity $O(\log^2 n)$ and with an approximation ratio of $1/2$.
The space complexity of the algorithm and the time
complexity of its preprocessing step is $O(n)$.
This also implies a $1/2$-approximation algorithm for finding
contiguous hotspots of whole trajectories in $O(n)$ time.
The low complexity of this algorithm, its practicable data structures,
and small approximation factor makes this algorithm suitable
for large real world trajectory data sets.

The rest of this paper is organized as follows.  In Section~\ref{sprel}
we describe the notation used in this paper and in Section~\ref{salg}
we present our algorithm.  Finally, in Section~\ref{scon} we conclude
this paper and mention possible directions for further studies.

\section{Preliminaries and Notation}
\label{sprel}
Let $T$ be a polygonal trajectory, which describes the location
of a moving entity through a specific time interval.
We use $T(t)$ to denote the entity's location at time $t$
in trajectory $T$.
In polygonal trajectories, the location of the entity is
recorded as different points in time; these we call the
vertices of the trajectory.  These vertices are linearly
interpolated to decide the location of the entity between
two contiguous vertices.  The sub-trajectory that connects
any two contiguous vertices of the trajectory are its edges.

For each vertex $v$ of trajectory $T$,
let $\mathrm{loc}(v)$ denote its location and
$\mathrm{tstamp}(v)$ denote its time-stamp.
With slight abuse of notation, we use vertices and
time-stamps interchangeably.  Thus, $u < v$ for vertices
$u$ and $v$ means $\mathrm{tstamp}(u) < \mathrm{tstamp}(v)$.
For two vertices or time-stamps $u$ and $v$ of trajectory $T$,
$T_{uv}$ denotes the sub-trajectory from $u$ to $v$.

For a square $r$ of fixed side length $s$, the weight of $r$ with respect
to trajectory $T$ is the maximum duration of a contiguous
sub-trajectory of $T$ contained in $r$.
A hotspot of trajectory $T$ is an axis-parallel square of
fixed side length $s$ with the maximum weight.
Every square discussed in this paper has fixed side length $s$
and is axis-parallel.  We may not mention these constraints
explicitly hereafter.

\section{An Algorithm for Answering Time-Windowed Queries}
\label{salg}
To improve the readability, we first present the algorithm with
the assumption that queries start and end with the time-stamp of
a trajectory vertex and prove its approximation factor.
We then discuss how to construct the data structures required in the algorithm.
Finally, we handle general queries that may start or end at a time
different from the time-stamps of all trajectory vertices and
show that this preserves the approximation factor.

\subsection{The Main Algorithm}
\label{sbase}
For any trajectory vertex $v$,
we define $\mathrm{hotend}_- (v)$ to denote the start of a
sub-trajectory of $T$ ending at $v$ with the maximum duration
that can be contained in a square,
$\mathrm{hotsquare}_- (v)$ to denote one such square, and
$\mathrm{hotdur}_- (v)$ to denote its duration.
Similarly,
$\mathrm{hotend}_+ (v)$ denotes the end of a sub-trajectory
of $T$ starting at $v$ with the maximum duration
that can be contained in a square,
$\mathrm{hotsquare}_+ (v)$ denotes one such square, and
$\mathrm{hotdur}_+ (v)$ denotes its duration.
The implementation of these functions is described in
Section~\ref{sprep}.

We also define $\mathrm{hotdur}_- (u, v)$ as the maximum value of
$\mathrm{hotdur}_- (w)$ for all vertices like $w$ in $T_{uv}$,
$\mathrm{hotsquare}_- (u, v)$ as its corresponding square,
and $\mathrm{hotend}_- (u, v)$ as its corresponding starting vertex,
in which $u$ and $v$ are two trajectory vertices.
We define $\mathrm{hotdur}_+ (u, v)$,
$\mathrm{hotsquare}_+ (u, v)$,
and $\mathrm{hotend}_+ (u, v)$ similarly.
The implementation of these functions are also explained in
Section~\ref{sprep}.

We now describe the algorithm for answering query $(u, v)$,
to find an approximate hotspot of sub-trajectory $T_{uv}$.
As mentioned before, here we assume that both $u$ and $v$ are
trajectory vertices.

\begin{enumerate}
\item
If $uv$ is an edge of the trajectory, it is trivial to find
its hotspot by considering the largest portion of the edge $uv$
that can fit in a square.
\item
Let $w$ be a trajectory vertex between $u$ and $v$ such that the
number of vertices in sub-trajectories $T_{uw}$ and $T_{wv}$ differ
by at most one.
A binary search between the vertices of $T_{uv}$ can find $w$.
Note that the duration of sub-trajectories $T_{uw}$ and $T_{wv}$ may
differ greatly.
\item
Let $u'$ be $\mathrm{hotend}_- (w, v)$ and
let $v'$ be $\mathrm{hotend}_+ (u, w)$.
There are three cases to consider.
\begin{enumerate}
\item
If $u' < u$ and $v < v'$,
$\mathrm{hotsquare}_- (w, v)$ contains $T_{uw}$ and
$\mathrm{hotsquare}_+ (u, w)$ contains $T_{wv}$;
let $r$ be $\mathrm{hotsquare}_- (w, v)$ if the duration
of $T_{uw}$ is greater than that of $T_{wv}$
and $\mathrm{hotsquare}_+ (u, w)$, otherwise.
Given that the duration of $T_{uv}$ is the sum of
the durations of $T_{uw}$ and $T_{wv}$,
the weight of $r$ is at least half of the weight of
the hotspot of $T_{uv}$.
\item
Now suppose $u < u'$ and $v' < v$.
Any hotspot of $T_{uv}$ either starts at a vertex of $T_{uw}$
or ends at a vertex of $T_{wv}$.  Consider the first case:
if a hotspot of $T_{uv}$ starts at a vertex of $T_{uw}$,
its weight should be equal to $\mathrm{hotdur}_+ (u, w)$.
For the second case, we can similarly argue that the weight of
a hotspot that ends at a vertex of $T_{wv}$ is $\mathrm{hotdur}_- (w, v)$.
Therefore, we can return either $\mathrm{hotsquare}_+ (u, w)$
or $\mathrm{hotsquare}_- (w, v)$ based on the relative values of
$\mathrm{hotdur}_+ (u, w)$ and $\mathrm{hotdur}_- (w, v)$.
\item
The remaining case is when $u' < u$ and $v' < v$ (the case when
$u < u'$ and $v' < v$ is similar and is omitted for brevity).
Again, any hotspot of $T_{uv}$ either starts at a vertex of $T_{uw}$
or ends at a vertex of $T_{wv}$.
If a hotspot of $T_{uv}$ starts at a vertex of $T_{uw}$,
its weight equals $\mathrm{hotdur}_+ (u, w)$.
Otherwise, the hotspot should start and end at two vertices of $T_{wv}$.
We perform this algorithm recursively for the interval $(w, v)$ to
find the approximate hotspot $r$ for this interval.
We can return either $\mathrm{hotend}_+ (u, w)$ or $r$ based on
their weight.
\end{enumerate}
\end{enumerate}

The time complexity and approximation factor of this algorithm
is shown in Theorem~\ref{tbase}.

\begin{thm}
\label{tbase}
Given a trajectory $T$ and after some preprocessing for
computing the functions $\mathrm{hotdur}$, $\mathrm{hotsquare}$,
and $\mathrm{hotend}$,
for each query $(u, v)$ we can find a square, whose weight is
at least half of the weight of the hotspot of $T_{uv}$,
with the time complexity $O(\log^2 n)$.
\end{thm}
\begin{proof}
The correctness of the algorithm is explained in its description.
For the approximation ratio, consider the tree formed by the
recursive invocations of the algorithm for a query.
The steps of the algorithm that return an approximate result
are 3.a, and 3.b, and the first case of 3.c, which appear only once
and as leaves in this tree.
Therefore, the weight of the square returned by the algorithm
is at least half of the weight of a hotspot.

For the time complexity, note that the algorithm is invoked
recursively only once in step 3.c and for a query containing half as
many points as the original query.
In each invocation, $\mathrm{hotend}$ and $\mathrm{hotdur}$ functions,
the time complexity of both of which is $O(1)$,
are called a constant number of times.
Therefore, the time complexity of answering a query containing
$m$ vertices is $T(m) \le T(m / 2) + O(\log m)$ (the $O(\log m)$
term is for finding $w$ in step 2).
Solving this recurrence yields $T(n) = O(\log^2 n)$ as required.
\end{proof}

\subsection{Preprocessing}
\label{sprep}
We next show how to implement functions
$\mathrm{hotend}$, $\mathrm{hotdur}$, and $\mathrm{hotsquare}$ for
single vertices.  The multi-vertex version of these functions
can be implemented using Range Minimum Query (RMQ) data structures.
RMQ data structures support finding the minimum (or the maximum) of
any contiguous interval of a sequence.
Using Cartesian trees, RMQ can be implemented with linear space and
$O(1)$ query complexity (for details, consult \cite{bender00}).

In the following algorithm, we use MinQueue data structure,
supporting the following operations: $\textrm{Insert}$ for inserting an item,
$\textrm{Remove}$ for removing the oldest item, and $\textrm{Min}$
for finding the item with the minimum value in the queue (MaxQueue data
structure is similar with a $\textrm{Max}$ operation instead).
There exists a clever implementation of MinQueue (and MaxQueue) using
two stacks, in which the time complexity of all three operations
is $O(1)$.

We now describe how to compute $\mathrm{hotdur}_-(v)$.
$\mathrm{hotend}_-(v)$ and $\mathrm{hotsquare}_-(v)$ can
be computed in parallel but are omitted for brevity.
Also note that $\mathrm{hotdur}_+(v)$, $\mathrm{hotend}_+(v)$,
and $\mathrm{hotsquare}_+(v)$ can be implemented similarly
by negating the time-stamps of the vertices.
Suppose $\textit{minx}$ and $\textit{miny}$ are instances of MinQueue
and $\textit{maxx}$ and $\textit{maxy}$ are instances MaxQueue and are initially empty.
The following steps are performed for every vertex of $T$
ordered by their time-stamps.
\begin{enumerate}
\item
Define $(x, y)$ as $\textrm{loc}(v)$.
Insert $v$ with the value $x$ into $\textit{minx}$ and $\textit{maxx}$.
Insert $v$ with the value $y$ into $\textit{miny}$ and $\textit{maxy}$.
\item
Repeatedly remove the oldest items from each of the four queues,
until both $\textrm{Max}(\textit{maxx}) - \textrm{Min}(\textit{minx})$ and
$\textrm{Max}(\textit{maxy}) - \textrm{Min}(\textit{miny})$ are at most $s$.
Defining $u$ as the oldest item in any of the queues,
the sub-trajectory $T_{uv}$ is the longest that
ends at vertex $v$, starts at a vertex of $T$, and
can be contained in a square.
\item
Let $u'$ be the vertex before $u$ in $T$ (the last vertex
removed from the queues).
Based on the condition in the previous step,
$T_{u'v}$ cannot be contained in a square but $T_{uv}$ can be.
To find $\textrm{hotend}_-(v)$, we need to find $p$,
the earliest point on the edge $u'u$ such that $T_{pv}$
can be contained in a square.
To do so, consider the four ways of aligning a corner of a square
with a corresponding corner of the bounding box of $T_{uv}$,
and choose the one that contains the longest portion of $u'u$.
\end{enumerate}

\begin{thm}
\label{tprep}
The time and space complexity of the preprocessing step for
functions $\mathrm{hotdur}$, $\mathrm{hotsquare}$,
and $\mathrm{hotend}$ is $O(n)$.
\end{thm}
\begin{proof}
All steps of the preprocessing perform $O(1)$ computation for
each vertex except step 2, which may extract many items from
the queues.  However, since only $n$ items are inserted into
each queue and each item can be removed at most once,
$O(n)$ items are removed from the queues during the whole algorithm.
Therefore, the time complexity of the algorithm is $O(n)$.
Since the Cartesian tree-based RMQ implementation of the
multi-vertex version of the functions has linear space and time
complexity, the time and space complexity of the preprocessing
is thus likewise linear.
\end{proof}

\subsection{Handling General Queries}
Theorem~\ref{tnonaligned} shows how to handle queries
whose start or end does not coincide with the time-stamp
of a trajectory vertex.
\begin{thm}
\label{tnonaligned}
After $O(n)$ preprocessing for a trajectory $T$ with $n$ vertices,
time-windowed hotspot queries can be answered approximately
with the time complexity $O(\log^2 n)$.
Each query is specified as two time-stamps $x$ and $y$ ($x < y$),
which may not coincide with the time-stamp of any trajectory vertex.
The algorithm returns a square whose weight is
at least half of the weight of the hotspots of $T_{xy}$.
\end{thm}
\begin{proof}
For the query $(x, y)$,
let $u$ be the first vertex on the trajectory at or after $x$ and $v$ be
the first vertex at or before $y$ (they can be found using binary
search on the vertices of $T$ in $O(\log n)$).

If the query is totally contained in a trajectory edge,
a hotspot can be found trivially by finding a square that
contains the longest portion of the edge.
Otherwise, suppose square $r$ is a hotspot of $T_{xy}$ and let $x'$
and $y'$ denote the start and end of a sub-trajectory of $T_{xy}$
contained in $r$ with the maximum duration.
There are two cases to consider.
If $u \le x'$ and $y' \le v$, based on Theorem~\ref{tbase} we can
find an approximate hotspot with an approximation factor of $1/2$.
Otherwise, suppose $x' < u$ (handling the case $v < y'$ is
similar and is omitted).
Suppose $T_{x'y'}$ contains $u$ (otherwise we can move
$r$ towards $u$ without changing its weight, since the $r$ is on a single edge).
Clearly, the weight of $r$ equals the sum of the
durations of $T_{x'u}$ and $T_{uy'}$;
the former is at most $\mathrm{hotdur}_- (u)$ and the latter
is at most $\mathrm{hotdur}_+ (u)$.
Therefore, the weight of either $\mathrm{hotsquare}_- (u)$
or $\mathrm{hotsquare}_+ (u)$ with respect to $T_{xy}$ is at
least half of the weight of $r$.
\end{proof}

\section{Concluding Remarks}
\label{scon}
The algorithm presented in this paper is very fast,
even for finding an approximate hotspot of the whole trajectory.
Querying the whole trajectory after preprocessing yields
Corollary~\ref{cwhole}.
\begin{cor}
\label{cwhole}
An approximate contiguous hotspot of a trajectory can be
found with the time complexity $O(n)$ and an approximation ratio of $1/2$.
\end{cor}
Several related problems seem interesting for further investigation.
It may be possible to include the side length of the hotspot $s$ in the query
by returning a sequence from the functions introduced in Section~\ref{sbase};
an algorithm to answer these extended queries would be very interesting.
The approximation ratio, although very good, may be improved.
This looks very important, especially from a practical point of view,
for querying large trajectory data sets.
Also, it seem interesting to answer time-windowed queries of
non-contiguous hotspots (see \cite{gudmundsson13} for more information).

\section*{Acknowledgements}
The author wishes to thank Dal for discussing this problem in Challenging
Thursdays 28.

\bibliographystyle{unsrt}

\begin{thebibliography}{10}

\bibitem{zheng15}
Y.~Zheng.
\newblock Trajectory data mining - an overview.
\newblock {\em ACM Transactions on Intelligent Systems and Technology},
  6(3):29:1--29:41, 2015.

\bibitem{li08}
Q.~Li, Y.~Zheng, X.~Xie, Y.~Chen, W.~Liu, and W.-Y. Ma.
\newblock Mining user similarity based on location history.
\newblock In {\em GIS}, page~34, 2008.

\bibitem{yuan11}
J.~Yuan, Y.~Zheng, L.~Zhang, X.~Xie, and G.~Sun.
\newblock Where to find my next passenger.
\newblock In {\em UbiComp}, pages 109--118, 2011.

\bibitem{yuan15}
N.~J. Yuan, Y.~Zheng, X.~Xie, Y.~Wang, K.~Zheng, and H.~Xiong.
\newblock Discovering urban functional zones using latent activity
  trajectories.
\newblock {\em IEEE Transactions on Knowledge and Data Engineering},
  27(3):712--725, 2015.

\bibitem{benkert10}
M.~Benkert, B.~Djordjevic, J.~Gudmundsson, and T.~Wolle.
\newblock Finding popular places.
\newblock {\em International Journal of Computational Geometry and
  Applications}, 20(1):19--42, 2010.

\bibitem{gudmundsson13}
J.~Gudmundsson, M.~J. van Kreveld, and F.~Staals.
\newblock Algorithms for hotspot computation on trajectory data.
\newblock In {\em SIGSPATIAL/GIS}, pages 134--143, 2013.

\bibitem{bannister14}
M.~J. Bannister, W.~E. Devanny, M.~T. Goodrich, J.~A. Simons, and L.~Trott.
\newblock Windows into geometric events: Data structures for time-windowed
  querying of temporal point sets.
\newblock In {\em The Canadian Conference on Computational Geometry}, pages
  11--19, 2014.

\bibitem{chan15}
T.~M. Chan and S.~Pratt.
\newblock Time-windowed closest pair.
\newblock In {\em The Canadian Conference on Computational Geometry}, pages
  141--144, 2015.

\bibitem{bokal15}
D.~Bokal, S.~Cabello, and D.~Eppstein.
\newblock Finding all maximal subsequences with hereditary properties.
\newblock In {\em Symposium on Computational Geometry}, pages 240--254, 2015.

\bibitem{chan16}
T.~M. Chan and S.~Pratt.
\newblock Two approaches to building time-windowed geometric data structures.
\newblock In {\em Symposium on Computational Geometry}, pages 28:1--28:15,
  2016.

\bibitem{bender00}
M.~A. Bender and M.~Farach-Colton.
\newblock The {LCA} problem revisited.
\newblock In {\em {LATIN}}, pages 88--94, 2000.

\end{thebibliography}

\end{document}